\documentclass[conference,10pt]{IEEEtran}
\usepackage{amsthm}
\usepackage{amsmath}
\usepackage{amsfonts,amssymb}
\usepackage{stmaryrd}
\usepackage[usenames,dvipsnames]{xcolor}
\usepackage{multirow}
\usepackage{longtable}
\usepackage[all,color, pdftex]{xy}
\usepackage{enumerate}
\usepackage{cases}
\usepackage{color}
\usepackage{graphicx}
\usepackage{caption}
\usepackage{subcaption}

\newtheorem{theorem}{Theorem}[section]
\newtheorem{corollary}[theorem]{Corollary}
\newtheorem{lemma}[theorem]{Lemma}

\theoremstyle{plain}

\theoremstyle{remark}
\newtheorem{definition}{Definition}
\newtheorem{example}{Example}

\newcommand{\A}{{\mathcal A}}
\newcommand{\B}{{\mathcal B}}
\newcommand{\C}{{\mathcal C}}
\newcommand{\Code}{{\mathcal C}}

\newcommand{\cR}{{\mathcal R}}
\newcommand{\sAA}{{\mathbb A}}
\newcommand{\sBB}{{\mathbb B}}
\newcommand{\sCC}{{\mathbb C}}
\newcommand{\nn}{{\mathbb N}}
\newcommand{\ff}{{\mathbb F}}

\newcommand{\bldc}{{\mbox{\boldmath $c$}}}
\newcommand{\bldx}{{\mbox{\boldmath $x$}}}
\newcommand{\bldy}{{\mbox{\boldmath $y$}}}

\newcommand{\bldsigma}{{\mbox{\boldmath $\sigma$}}}

\begin{document}

\sloppy

\title{Bounds for Batch Codes with Restricted Query Size}
\author{
\IEEEauthorblockN{\bf Hui Zhang}%
\IEEEauthorblockA{Department of Computer Science\\
Technion -- Israel Institute of Technology\\
Haifa 32000, Israel}%
\and%
\IEEEauthorblockN{\bf Vitaly Skachek}%
\IEEEauthorblockA{Institute of Computer Science\\
University of Tartu\\ 
Tartu 50409, Estonia}%
}

\maketitle
\footnotetext[1]{This work was done while the first author was with the Institute of Computer Science, University of Tartu, Tartu 50409, Estonia. This work is supported in part by the grants PUT405 and IUT2-1 from the Estonian Research Council and by the EU COST Action IC1104.}

\begin{abstract}
We present new upper bounds on the parameters of batch codes with restricted query size. These bounds are an improvement on the Singleton bound. The techniques for derivations of these bounds are based on the ideas in the literature for codes with locality. By employing additional ideas, we obtain further improvements, 
which are specific for batch codes. 
\end{abstract}

\section{Introduction}
Batch codes were originally proposed in~\cite{Ishai} for load balancing in the distributed server systems. They can also potentially be used in the application for private information retrieval~\cite{Vardy}. 
Batch codes are also known as \emph{switch codes}, and they were studied in~\cite{Chee, WKC2015, Wang2013} in the context of network switches. 
Combinatorial batch codes were studied, for example, in~\cite{Bhattacharya,Brualdi,Bujtas,Gal}.  
Several constructions of families of batch codes using graph-theoretic tools were recently presented in~\cite{dimakis-batch}. 

Batch codes have a lot of similarities with so-called \emph{locally-repairable} codes, or \emph{codes with locality}, which are of potential use in the distributed storage systems~\cite{Cadambe, dimakis-survey, FY2014,Yekhanin2012,Prakash, RPDV2014,RMV2014, RMV2014-journal, Tamo}. 

This work continues the line of research that was started in~\cite{Yekhanin2012}, where the bounds on the parameters of codes with locality of a single symbol, which improve on the Singleton bound, were presented. Further, the authors of~\cite{RPDV2014} generalize that approach towards codes \emph{with locality and availability}, where every symbol can be recovered from several recovery sets. This work was further generalized to codes, which allow for cooperative recovery of several symbols~\cite{RMV2014}. A variation of that model, where several (possibly different) symbols are recovered in parallel, was studied in~\cite{Paudyal}.   

In this work, by building on the ideas in~\cite{Yekhanin2012,RPDV2014,RMV2014}, we derive new upper bounds on the parameters of batch codes. 
The resulting bounds are an improvement on the classical Singleton bound, and they do not depend on the size of the underlying alphabet. 
By using some additional ideas, we further improve the resulting bound. 

\section{Notation}

Throughout the paper, we denote by $\nn$ the set of nonnegative integers. For $n \in \nn$, we denote $[n] \triangleq \{1, 2, \cdots, n\}$.

We start with the definition of a new family of $(k,n,r,t)$-batch codes with restricted query size, which is a variation of the definition of \emph{primitive batch codes} in~\cite{Ishai}.
\begin{definition}
A {\em primitive $(k,n,r,t)$ batch code $\Code$ with restricted query size} over an alphabet $\Sigma$ encodes a string $\bldx\in\Sigma^k$ into a string $\bldy=\Code(\bldx)\in\Sigma^n$, such that for all multisets of indices $\{i_1,i_2, \dots,i_t\}$, where all $i_j \in [k]$, each of the entries $x_{i_1},x_{i_2},\dots,x_{i_t}$ can be retrieved independently of each other by reading at most $r$ symbols of $\bldy$. It is assumed that the symbols used to retrieve each of the variables $x_{i_j}$, for $1\leq j \leq t$, are all disjoint.
\label{def:batch-code}
\end{definition}
This definition is different from that of the primitive batch codes in~\cite{Ishai} by adding an additional restriction that each symbol is recovered by reading 
at most $r$ symbols of $\bldy$. In sequel, we will simply call the codes in Definition~\ref{def:batch-code} the $(k,n,r,t)$ \emph{batch codes}.

We will use the notation $d$ for the minimum (Hamming) distance of the $(k,n,r,t)$ batch code $\Code \triangleq \{\bldy \in \Code(\bldx) \; : \; \bldx \in \Sigma^k\}$. 

If the code alphabet $\Sigma$ is a finite field $\ff$, and for all $\bldx \in \ff^k$, $\Code(\bldx)$ is a linear transformation, then 
the corresponding batch code is \emph{linear}. Otherwise, the batch code is \emph{non-linear}. 

The following simple result was observed in~\cite{Lipmaa} for binary linear codes and in~\cite{Zumbraegel} for non-linear codes over general $\Sigma$.
\begin{lemma}
If a $(k,n,r,t)$ batch code has the minimum distance $d$, then $d\geq t$.
\label{lem:min-distance}
\end{lemma}

\begin{proof}
For any two strings $\bldx_1,\bldx_2\in\Sigma^k$, there is at least one coordinate, where they are different, denote it by $i$. 
Then, in the codewords $\bldy_1= \Code (\bldx_1),\bldy_2= \Code (\bldx_2)$, there exist at least $t$ coordinates, where they are different. Otherwise, it is not possible to retrieve the $t$-tuple of coordinates $\{i,i,\dots,i\}$.
\end{proof}

The $(k,n,r,t)$ batch code satisfies the Singleton bound, that is, $d\leq n-k+1$. By Lemma~\ref{lem:min-distance}, the rate of the code is 
$$\frac{k}{n}\leq \frac{k}{k+d-1}\leq \frac{k}{k+t-1} \; .$$
\medskip 

In this paper, we present a new upper bound on the parameters of a linear batch code, when the query size is restricted. 
The primary technique is based on that used for derivation of bounds on the distance of locally recoverable codes in the series of 
works~\cite{Yekhanin2012,RPDV2014,RMV2014}. By using some additional ideas, we are able to further improve the obtained bound.

\section{Upper Bound on the Parameters of Batch Codes}

For a code $\C\subseteq\Sigma^n$, $S\subseteq [n]$ and $\bldy \in \C$, let $\bldy|_S\subseteq\Sigma^{|S|}$ denote the subword of $\bldy$ composed of the symbols indexed by $S$.

In this section, we present an algorithm akin to that in \cite{Yekhanin2012,RPDV2014,RMV2014}, which allows to derive an upper bound for the minimum distance of a batch code. Throughout the paper, we consider linear $(k,n,r,t)$ batch codes over the field $\ff$, whose size is an arbitrary prime power.

We will use the following auxiliary results. 

\begin{lemma}
Let $\Code$ be a linear $(k,n,r,t)$ batch code over the finite field $\ff$, $\bldx \in \ff^k$, $\bldy = \C(\bldx)$, 
and assume that the set of coordinates of $\bldy$ indexed by $S \subseteq [n]$ is used to recover $x_i$ for some $i$. 
Then, there exists $\ell \in S$, such that if we fix the values of $x_i$ and of $\bldy|_{S \backslash \{ \ell \}}$, 
the value of $y_\ell$ will be uniquely determined. 
\label{lemma:sets}
\end{lemma}

\begin{proof}
\begin{enumerate}
\item
Pick a random $\ell \in S$ and assume that $y_\ell$ is not uniquely determined given the values of $x_i$ and of $\bldy|_{S \backslash \{ \ell \}}$. 
Then, there exist two words $\bldy_1 = \Code(\bldx_1)$ and $\bldy_2 = \Code(\bldx_2)$, such that $\bldx_1$ and $\bldx_2$ coincide on coordinate $i$, 
and $\bldy_1$ and $\bldy_2$ coincide on the coordinates indexed by 
$S \backslash \{ \ell \}$, but differ on the coordinate $\ell$. Consider the word $\bldy_1 - \bldy_2 = \Code(\bldx_1 - \bldx_2)$. 
It has zeros in all coordinates indexed by $S \backslash \{ \ell \}$, yet its $\ell$'s coordinate is nonzero. 
Additionally, the $i$-th coordinate of $\bldx_1 - \bldx_2$ is zero. 

Next, take an arbitrary codeword $\bldc \in \Code$. For any $\alpha \in \ff$, $\alpha \neq 0$, 
the word $\bldc + \alpha(\bldy_1 - \bldy_2)$ differs from $\bldc$ in coordinate $\ell$ (the difference $\alpha(\bldy_1 - \bldy_2)$ 
can take any nonzero value in $\ff$ for different nonzero values of $\alpha$), yet 
the coordinates in $S \backslash \{ \ell \}$ and the $i$-th coordinate in the corresponding information words are identical. Therefore, 
for any $\bldc \in \Code$, the bit $y_\ell$ does not effect the value of $x_i$ and so it can be ignored. We obtain that $y_\ell$ is
not helpful for recovery of $x_i$, and the set $S \backslash \{ \ell \}$ is sufficient for its retrieval. 
\item
We showed that if there exists $\ell \in S$ such that $y_\ell$ is not uniquely determined, then 
the set $S \backslash \{ \ell \}$ is sufficient for the retrieval. By repeating this argument, we end up with the minimal (nonempty) retrieval set $S_0$ of $x_i$. 
Therefore, every symbol in $\bldy|_{S_0}$ is uniquely determined. 
\end{enumerate}
\end{proof}

The following corollary follows directly from Lemma~\ref{lemma:sets}. 
\begin{corollary}
Let $\Code$ be a linear $(k,n,r,t)$ batch code over $\ff$, $\bldx \in \ff^k$, $\bldy = \C(\bldx)$. Let $S_1, S_2, \cdots, S_t \subseteq [n]$ be $t$ disjoint recovery sets for the coordinate $x_i$. Then, there exist indices $\ell_2 \in S_2$, $\ell_3 \in S_3$, $\cdots$, $\ell_t \in S_t$, such that 
if we fix the values of all coordinates of $\bldy$ indexed by the sets $S_1, S_2 \backslash \{ \ell_2 \}, S_3 \backslash \{ \ell_3 \}, \cdots,  S_t \backslash \{ \ell_t \}$, then the values of the coordinates of $\bldy$ indexed by $\{\ell_2, \ell_3, \cdots, \ell_t \}$ are uniquely determined. 
\label{cor:sets}
\end{corollary}

The following theorem is the main result in this section. 

\begin{theorem}
\label{thrm:main} 
Let $\C$ be a linear $(k,n,r,t)$ batch code over $\ff$ with the minimum distance $d$. Then, 
\begin{equation}
d \leq n-k-(t-1)\left(\left\lceil\frac{k}{rt-t+1}\right\rceil-1\right)+1 \; . 
\label{eq:main-statement}
\end{equation}
\end{theorem}

\begin{proof}
Apply algorithm in Figure~\ref{algo} to the linear $(k,n,r,t)$ batch code $\C$.

\begin{figure}[h!!]
\centering\footnotesize
\begin{tabular}{l}
\hline
\hline
$\,$ \\
{\bf Input:} linear $(k,n,r,t)$ batch code $\Code$ \\
1: $\C_0=\C$ \\
2: $j=0$ \\
3: while $|\C_j|>1$ do \\
4: $j=j+1$ \\
5: Choose the multiset $\{i_j^1,i_j^2,\dots,i_j^t\}\subseteq[k]$ and disjoint subsets \\ 
~~~~$S_{j}^1,\dots,S_{j}^t\in[n]$, where $S_j^\ell$ is a recovery set for the information \\
~~~~bit $i_j^\ell$, such that there exist at least two codewords in $\C_{j-1}$ \\
~~~~that differ in (at least) one coordinate\\
6: Let $\bldsigma_j\in\Sigma^{|S_j|}$ be the most frequent element in the multiset \\
~~~$\{\bldx|_{S_j}:\bldx\in\C_{j-1}\}$, where $S_j=S_{j}^1\cup\dots\cup S_{j}^t$ \\
7: Define $\C_j \triangleq \{\bldx:\bldx\in\C_{j-1},\bldx|_{S_j}=\bldsigma_j\}$ \\
8: end while \\
{\bf Output:} $\C_{j-1}$ \\
$\,$ \\
\hline
\hline
\end{tabular}
\caption{Algorithm for constructing a subcode}
\label{algo}
\end{figure}

First, we prove that the algorithm is well-defined, that is, we are able to choose $i_j^1,i_j^2,\dots,i_j^t$ and $S_{j}^1, S_{j}^2,\dots,S_{j}^t$ in line 5. Let 
$$\Gamma_j=\bigcup_{j'\in [j]}S_{j'}\text{~~and~~}\Lambda_j=\bigcup_{j'\in [j]}\{i_{j'}^1,i_{j'}^2,\dots,i_{j'}^t\} \; .$$ 
By condition in line 3, we have $|\C_{j-1}|>1$, thus implying that there is at least one coordinate such that two codewords $\bldy_1$ and $\bldy_2$ in $\C_{j-1}$ are not equal. Let $\bldx_1$ and $\bldx_2$ be information words, corresponding to $\bldy_1$ and $\bldy_2$, respectively. It follows that $\bldx_1 \neq \bldx_2$, 
and there is at least one coordinate, say $x_\ell$, such that $\bldx_1$ and $\bldx_2$ disagree in that coordinate. 
Then, in Step 5 of the algorithm in Figure~\ref{algo}, we can use the multiset of indices $\{\ell, \ell, \cdots, \ell\}$ and the corresponding recovery sets, 
and this choice is feasible. 


Assume that the algorithm terminates with $j=\tau+1$, and its output is a code $\C_\tau\subseteq\C$. Let $\A_j=\Gamma_j\setminus\Gamma_{j-1}$ and $a_j=|\A_j|$. Then $a_j\leq t \cdot r$. Define the multiset $\B_j \triangleq \Lambda_j \setminus \Lambda_{j-1}$, $\lambda_j=|\B_j|$. Let $\mu_j$ be the number of different symbols in $\B_j$. Here, $1 \leq \mu_j\leq \lambda_j$.

It follows from Corollary~\ref{cor:sets} that if we are to recover in Step 5 the total of $t$ copies of $\mu_j$ different bits, then 
in each of the $t - \mu_j$ sets will be at least one bit fixed. Moreover, 
$$
|\C_{j-1}|\leq |\C_j|\cdot{q^{a_j-(t-\mu_j)}} \; , 
$$
and, therefore,  
\begin{equation}
|\C_j|\geq |\C_{j-1}|/{q^{a_j-(t-\mu_j)}} \; .
\label{eq:decrease-code-size}
\end{equation} 


Next, we bijectively map $\C_\tau$ to a new code $\C'\subseteq
\ff^{n-|\Gamma_\tau|}$ by deleting the coordinates in $\Gamma_\tau$. The minimum distance of $\C'$ is at least $d$. 
Then,
\begin{align*}
\log_q|\C'| = \log_q|\C_\tau| & \geq \log_q|\C|-\sum_{j=1}^\tau(a_j-t+\mu_j) \\
 & \geq k-|\Gamma_\tau|+\sum_{j=1}^\tau(t-\mu_j) \; , 
\end{align*}
where the penultimate transition is due to repeated applications of~(\ref{eq:decrease-code-size}). 
We also have
\begin{align}
0=\log_q|\C_{\tau+1}| & \geq \log_q|\C|-\sum_{j=1}^{\tau+1}(a_j-t+\mu_j) \nonumber \\
 & \geq k-\sum_{j=1}^{\tau+1}(rt-t+\mu_j) \; .
\label{eq:mu-bound}
\end{align}
By applying the Singleton bound to $\C'$, we obtain
\begin{eqnarray}
d & \leq & (n-|\Gamma_\tau|)-(k-|\Gamma_\tau|+\sum_{j=1}^\tau(t-\mu_j))+1 \nonumber \\
& = & n-k+1-\sum_{j=1}^\tau(t-\mu_j) \; . 
\label{d}
\end{eqnarray}
The right-hand side of~(\ref{d}) decreases when each of $\mu_j$, $1 \leq j \leq \tau$, decreases. In order to minimize the bound, we choose $\mu_j=1$, for $1\leq j\leq \tau$. Hence,
\begin{align*}
d & \leq n-k+1-(t-1) \, \tau \\
 & \leq n-k+1-(t-1)\left(\left\lceil\frac{k}{rt-t+1}\right\rceil-1\right) \; ,
\end{align*}
where the last transition is due to~(\ref{eq:mu-bound}). 
\end{proof}

\begin{example}
Consider a binary $(k,n,1,t)$ batch code with minimum distance $d=t$ and length $n = k \cdot t$. 
Since each information symbol can be recovered from $t$ sets of size $1$, 
such code can be obtained by concatenation of the trivial binary code of length $k$ with binary repetition code of length $t$. 

Then,~(\ref{eq:main-statement}) is equivalent to 
$$
n \geq tk+d-t = tk \; , 
$$ 
and we observe that the bound in Theorem~\ref{thrm:main} is tight in that case.
\end{example}

\begin{corollary}
\label{coro1} Let $\C$ be a linear $(k,n,r,t)$ batch code over $\ff$ with the minimum distance $d$. Then,
\begin{align}
\label{bound} n\geq \max_{1\leq \beta\leq t, \beta \in \nn}\left\{(\beta-1)\left(\left\lceil\frac{k}{r\beta-\beta+1}\right\rceil-1\right)+k+d-1\right\}.
\end{align}
\end{corollary}

\begin{proof}
The code $\C$ is a $(k,n,r,\beta)$ batch code for all integers $\beta$, $1 \leq \beta \leq t$. Therefore, the claim follows from Theorem~\ref{thrm:main}.
\end{proof}

\begin{corollary}
\label{coro2} Let $\C$ be a linear \emph{systematic} $(k,n,r,t)$ batch code over $\ff$ with the minimum distance $d$. Then,
\begin{multline}
n\geq \max_{2\leq \beta\leq t, \beta \in \nn}\Bigg\{(\beta-1)\left(\left\lceil\frac{k}{r\beta-\beta-r+2}\right\rceil-1\right) \\ +k+d-1\Bigg\} \; .
\label{eq:systematic}
\end{multline}
\end{corollary}

\begin{proof}
If the batch code is systematic, then there exists a recovery set of size one for each information symbol. Then, in the proof of Theorem~\ref{thrm:main}, we have $a_j\leq|S_j|\leq rt-r+1$. We obtain:
$$
\tau\geq \left\lceil\frac{k}{rt-t-r+2}\right\rceil-1 \; , 
$$
and from~(\ref{d}), it follows that 
$$
d\leq n-k+1-(t-1) \left(\left\lceil\frac{k}{rt-t-r+2}\right\rceil-1\right) \; . 
$$ 
The claimed result follows by an argument similar to the one used in Corollary~\ref{coro1}.
\end{proof}

\medskip 

\begin{example} 
Take $r=2$ and $t=\beta=2$. Then by (\ref{eq:systematic}), $n\geq \left\lceil\frac{k}{2}\right\rceil+k+d-2$. This bound can be attained by the linear systematic codes of minimum distance 2 encoding $\bldx=\{x_i \; : \; 1\leq i\leq k\}$ into $\bldy=\{y_i \; : \; 1\leq i\leq n\}$, where: 
\begin{itemize}
\item $y_i=x_i$ for $1\leq i\leq k$, and $y_j=x_{2(j-k)-1}+x_{2(j-k)}$ for $k+1\leq j\leq k+k/2$, when $k$ is even,
\item $y_i=x_i$ for $1\leq i\leq k$, $y_j=x_{2(j-k)-1}+x_{2(j-k)}$ for $k+1\leq j\leq k+(k-1)/2$, and $y_{k+(k+1)/2}=x_k$, when $k$ is odd.
\end{itemize}
\end{example}

\medskip

\begin{example}
Consider batch codes, which are obtained by taking simplex codes as suggested in \cite{WKC2015}. 
It was shown therein that, for example, the linear code, formed by the generator matrix 
\[
\left(
\begin{matrix}
1&0&0&1&1&0&1 \\
0&1&0&1&0&1&1 \\
0&0&1&0&1&1&1
\end{matrix}
\right)
\]
is a $(3,7,2,4)$ batch code with the minimum distance $d=4$. Here $r=2$ and $t = 4$. 

Pick $\beta=2$. The right-hand side of equation~(\ref{eq:systematic}) can be re-written as 
\begin{multline*}
 (2-1)\left(\left\lceil\frac{3}{2\cdot 2- 2 - 2 + 2}\right\rceil-1\right) +3+ 4 - 1 = 7 \; ,
\end{multline*}
and therefore the bound in~(\ref{eq:systematic}) is attained with equality for the choice $\beta = 2$.  

Assume that the pair $\{x_1, x_1\}$ is to be recovered. The sets $\{x_1\}$ and $\{x_2,x_1+x_2\}$ are sufficient for this task. 
Then, $S_1=\{1,2,4\}$, and the most frequent element in Step 6 of the algorithm is $\bldsigma_1=(1,0,1)$. We obtain 
the code $\C_1$ with two codewords (when puncturing the coordinates in $S_1$): 
\begin{equation}
\C_1 = \left\{ \left( \, 0 \; 1 \; 0 \; 1 \, \right) , \; \left( \, 1 \; 0 \; 1 \; 0 \, \right) \right\} \; . 
\end{equation}
The code $\C_1$ is an MDS code of dimension $1$ and minimum distance $4$. 
\end{example}

\vskip 5pt
Next, we turn to the asymptotic analysis, when $n \rightarrow \infty$. 
Let the {\em relative rate} of the code $\Code$ be $R \triangleq k/n$, and the {\em relative minimum distance} be $\delta \triangleq d/n$. 
Assume that $r$ and $t$ are fixed. From~(\ref{bound}), by ignoring the $o(1)$ term, and by choosing the optimal value $\beta = t$, 
we have
\begin{equation}
\delta \le 1 - \left(\frac{rt}{rt - t + 1} \right)R \; . 
\label{eq:bound-asymptotic}
\end{equation}
A variation of this bound can be obtained by using the Plotkin bound instead of the Singleton bound in the proof of 
Theorem~\ref{thrm:main}.

Indeed, take the subcode $\C_j$ for some $j$ in the algorithm in Figure~\ref{algo}, and assume that $\mu_j = 1$. Denote the length of $\C_j$ by $n'=n-|\Gamma_j|$, and the dimension by $k'\geq k-|\Gamma_j|+(t-1)j$. Delete the coordinates $\Gamma_j$ and bijectively map $\C_j$ to a new code $\C'\subseteq \Sigma^{n-|\Gamma_j|}$, which has the minimum distance at least $d$. 

By applying the Plotkin bounds to $\C'$, we have: if $ d > (1-1/q)n'$, then $q^{k'}\leq\frac{qd}{qd-(q-1)n'}$, where $q$ is the underlying field size. 
Hence, if $d>(1-1/q)(n-|\Gamma_j|)$,
\begin{align}
\label{plotkin} q^{k-|\Gamma_j|+(t-1)j}\left(qd-(q-1)(n-|\Gamma_j|)\right)\leq qd.
\end{align}

When $j=0$,~(\ref{plotkin}) is essentially a classical Plotkin bound. Plotkin bound is tight, and in particular it is attained by the Simplex code. 

\section{Further Improvement}

In this section, by employing some additional ideas, we show an improvement on the bound~(\ref{bound}). 
Hereafter, we assume that $\mu_j=1$ for all $1\leq j\leq \tau$ (i.e. in each step $i$ of the algorithm, the set $S_i$  recovers multiple copies 
of one symbol). 
Additionally, we assume that 
\begin{equation}
k \geq 2(rt-t+1)+1 \; , 
\label{eq:k-range}
\end{equation}
which implies that $\tau \ge 2$. 

Let $\epsilon$ and $\lambda$ be some positive integers,
whose values will be established in the sequel. 
Consider the following three cases.
\begin{description}
\item[\bf Case 1: $\exists i,j \; : \; |S_i \cap S_j|\geq \epsilon$.] $\,$
\newline
W.l.o.g. assume $i=1$ and $j=2$ (the proof is similar for any choice of $i$ and $j$). Then, 
$a_2 \leq |S_2|-|S_1\cap S_2| \leq rt - \epsilon$. In that case, in equation~(\ref{eq:mu-bound}) in Theorem~\ref{thrm:main}, we obtain
\begin{align*}
0 = \log_q|\C_{\tau+1}| & \geq k-(rt-t+1)(\tau+1)+\epsilon \; , 
\end{align*}
which implies that  $\tau\geq \left\lceil\frac{k+\epsilon}{rt-t+1}\right\rceil-1$.
Hence,
\begin{align*}
d & \leq n-k+1-(t-1) \left( \left\lceil \frac{k+\epsilon}{rt-t+1} \right\rceil -1 \right) \; .
\end{align*}

\item[\bf Case 2: $\forall i,j \; : \; |S_i\cap S_j|<\epsilon$, and $\exists i \; : \; |S_i|\leq rt-\lambda$.] $\,$
\newline
W.l.o.g. assume $i=1$ (the proof is similar for any choice of $i$ and $j$). Then, $a_1\leq|S_1|\leq rt-\lambda$. 
In that case, in equation~(\ref{eq:mu-bound}) in Theorem~\ref{thrm:main}, we obtain
\begin{eqnarray*}
0 & = & \log_q|\C_{\tau+1}| \\
& \geq & k- \left( (rt-t+1)\tau+rt-\lambda-t+1 \right) \; ,  
\end{eqnarray*}
which implies that $\tau\geq \left\lceil\frac{k+\lambda}{rt-t+1}\right\rceil-1$.
Hence,
\begin{align*}
d & \leq n-k+1-(t-1) \left( \left\lceil\frac{k+\lambda}{rt-t+1}\right\rceil-1 \right) \; . 
\end{align*}

\item[\bf Case 3: $\forall i,j \, : \; |S_i\cap S_j|<\epsilon$, and $\forall i \; : \;  |S_i| > rt-\lambda$.] $\,$ \\
\noindent
In this case, akin to inclusion-exclusion principle, the total number of coordinates contained in any of the sets $S_i$, is given by:
\begin{eqnarray*}
\label{bound3} 
n & \geq & \sum_{i=1}^{k}|S_i|-{k\choose 2}(\epsilon-1) \\
& \geq & (rt-\lambda+1)k-{k\choose 2}(\epsilon-1) \; .
\end{eqnarray*}
\end{description}

Denote
\begin{eqnarray*}
\sAA & = & \sAA(k, r, d, \beta, \epsilon) \\
&& \triangleq (\beta-1)\left(\left\lceil\frac{k+\epsilon}{r\beta-\beta+1}\right\rceil-1 \right)+k+d-1 \; , \\
\sBB & = & \sBB(k, r, d, \beta, \lambda) \\
&& \triangleq (\beta-1)\left(\left\lceil\frac{k+\lambda}{r\beta-\beta+1}\right\rceil-1\right)+k+d-1 \; , \\
\sCC & = & \sCC(k, r, \beta, \lambda, \epsilon) \\
&& \triangleq (r\beta-\lambda+1)k-{k\choose 2}(\epsilon-1) \; .
\end{eqnarray*}

\begin{theorem}
Let $\C$ be a linear $(k,n,r,t)$ batch code with the minimum distance $d$. Then,
\begin{multline}
\label{boundn} n\geq \\
\max_{\beta\in \nn \cap \left[1,\min\left\{t,\big\lfloor\frac{k-3}{2(r-1)}\big\rfloor\right\}\right]} \left\{ \max_{\epsilon,\lambda\in\nn \cap[1,r\beta-\beta]} \left\{ \min \left\{\sAA,\sBB,\sCC\right\}\right\} \right\} \; . 
\end{multline}
\end{theorem}

\begin{proof}
The code $\C$ is a $(k,n,r,\beta)$ batch code for any integer $\beta$, $1 \le \beta \le t$. Apply the algorithm in Figure~\ref{algo} to the code $\C$.

For any fixed values $\epsilon$ and $\lambda$, one of the above three cases must occur. In the above analysis, we can choose any $\beta\in \nn \cap \left[1,\min\left\{t,\Big\lfloor\frac{k-3}{2(r-1)}\Big\rfloor\right\}\right]$ and any pair of $\epsilon,\lambda \in \nn \cap [1,r\beta-\beta]$. 

Here, the restriction $\beta \le \Big\lfloor\frac{k-3}{2(r-1)}\Big\rfloor$ is required in order to make sure that~(\ref{eq:k-range}) holds, thus implying
$\tau\geq 2$, and the condition $\epsilon,\lambda \in [1,r\beta-\beta]$ is required because $0\leq|S_i\cap S_j|\leq r\beta-\beta$ and $\beta\leq|S_i|\leq r\beta$.
\end{proof}

\begin{table*}[t]
\setlength{\tabcolsep}{12pt}
\def\arraystretch{1.3}
\begin{center}
\begin{tabular}{ r | c c c c c c c} 
\hline \hline
	\vspace{-1ex}
   & Singleton & $r=3$ & $r=3$ & $r=3$ & $r = 5$ & $r=5$ & $r=5$ \\  
	 & bound & $t=2$ & $t=3$ & $t=5$ & $t = 2$ & $t = 3$ & $t=5$ \\ \hline 
$R = 0.0$ &   1.0  & 1.0 & 1.0 & 1.0 & 1.0 & 1.0 & 1.0 \\ 
$R = 0.1$ &   0.9  & 0.88 & 0.87143 & 0.86364 & 0.8889 & 0.88462 & 0.88095\\ 
$R = 0.2$ &   0.8  & 0.76 & 0.74386 & 0.72727 & 0.7778 & 0.76923 & 0.76190 \\ 
$R = 0.3$ &   0.7  & 0.64 & 0.61429 & 0.59091 & 0.6667 & 0.65385 & 0.64286 \\ 
$R = 0.4$ &   0.6 & 0.52 & 0.48571 & 0.45455 & 0.5556 & 0.53846 & 0.52381\\ 
$R = 0.5$ &   0.5  & 0.4 & 0.35714 & 0.31818 & 0.4444 & 0.42308 & 0.40476 \\ 
$R = 0.6$ &   0.4  & 0.28 & 0.22857 & 0.18182 & 0.3333 & 0.30769 & 0.28571 \\ 
$R = 0.7$ &   0.3  & 0.16 & 0.1 & 0.045455 & 0.2222 & 0.19231 & 0.16667 \\ 
$R = 0.8$ &   0.2  & 0.04 & -- & -- & 0.1111 & 0.076923 & 0.047619 \\ 
$R = 0.9$ &   0.1  & -- & -- & -- & 0.0 & -- & -- \\ 
$R = 1.0$ &   0.0  & -- & -- & -- & -- & -- & -- \\ 
\hline \hline
\end{tabular}        
\end{center}
\caption{Bounds on the relative minimum distance $\delta$.}
\label{table}
\end{table*}
\medskip

Since $\epsilon \leq r\beta-\beta$, the bound (\ref{boundn}) is tighter than (\ref{bound}) when 
$$
\left\lceil\frac{k+\epsilon}{r\beta-\beta+1}\right\rceil=\left\lceil\frac{k+\lambda}{r\beta-\beta+1}\right\rceil \ge \left\lceil\frac{k}{r\beta-\beta+1}\right\rceil+1 \; , 
$$ 
and in addition $\sCC \geq \max\{\sAA, \sBB\}$, for some $\beta$, $\epsilon$ and $\lambda$.
\medskip

\begin{example}
Take $k=12$, $r=2$ and $t=3$. The maximum of the right-hand side of expression (\ref{bound}) is obtained when $\beta=3$. 
For that selection of parameters, we have  
$n\geq 15+d\geq 18$.

At the same time, by taking $\beta=3$, $\lambda = 1$ and $\epsilon=1$, we obtain from~(\ref{boundn}) that 
\[
\sAA = \sBB = 17 + d \; \mbox{ and } \; \sCC = 6 \cdot 12 - 0 = 72 \; , 
\]
and so 
\[
n\geq \min \{ 17+ d, 72\} \geq 20 \; . 
\]
\end{example}

\section{Numerical Results}

In this section, we consider the asymptotic regime, where $n \rightarrow \infty$, $R = k/n$ and $\delta = d/n$ are constants,
and $r$ and $t$ are fixed. Then, the bound~(\ref{bound}) can be rewritten as~(\ref{eq:bound-asymptotic}). 
 
Next, consider the bound~(\ref{boundn}). The expression $\sCC$ does not depend on $d$, 
and therefore the trade-off between $R$ and $\delta$ follows from $\sAA$ and $\sBB$, in which case it coincides with~(\ref{eq:bound-asymptotic}).  

In Table~\ref{table} we present some values of the pairs $(\cR, \delta)$ for the Singleton bound and for the bound~(\ref{eq:bound-asymptotic}).  
The entry is marked as `--' when no code with corresponding parameters exists.

\section*{Acknowledgement}

The authors wish to thank Helger Lipmaa and Jens Zumbr\"agel for helpful discussions.

\end{document}